\documentclass[a4paper,english,cleveref, autoref]{lipics-v2019}

\usepackage[linesnumbered,ruled]{algorithm2e}
\usepackage{booktabs} 
\usepackage{enumitem}

\newcommand{\X}{\textsf{X}} 

\newtheorem{fact}{Fact}

\newcommand{\E}{\textnormal{\textrm{E}}}
\newcommand{\HH}{\ensuremath{\mathcal{H}}}
\newcommand{\ip}[2]{\langle{#1},{#2}\rangle}
\newcommand{\dist}{\text{dist}}

\newcommand{\conf}[2]{#1}
\renewcommand{\conf}[2]{#2}

\usepackage[textsize=small]{todonotes}
\newcommand{\matodo}[1]{\todo[inline, color=green!30]{#1}}
\newcommand{\tctodo}[1]{\todo[inline, color=blue!30]{#1}}
\newcommand{\mvtodo}[1]{\todo[inline,color=yellow!30]{#1}}
\newcommand{\rptodo}[1]{\todo[inline,color=red!30]{#1}}

\renewcommand{\matodo}[1]{}
\renewcommand{\tctodo}[1]{}
\renewcommand{\mvtodo}[1]{}
\renewcommand{\rptodo}[1]{}

\usepackage{pgfplots}
\usepgfplotslibrary{groupplots,colorbrewer}
\pgfplotsset{
  legend style = {font=\ttfamily}
}
\pgfplotsset{
cycle list/Set1-5,
cycle multiindex* list={
mark list*\nextlist
Set1-5\nextlist
},
}

\nolinenumbers

\conf{
\EventEditors{Michael A. Bender, Ola Svensson, and Grzegorz Herman}
\EventNoEds{3}
\EventLongTitle{27th Annual European Symposium on Algorithms (ESA 2019)}
\EventShortTitle{ESA 2019}
\EventAcronym{ESA}
\EventYear{2019}
\EventDate{September 9--11, 2019}
\EventLocation{Munich/Garching, Germany}
\EventLogo{}
\SeriesVolume{144}
\ArticleNo{8}
}
{\hideLIPIcs}

\title{PUFFINN: Parameterless and Universally Fast FInding of Nearest Neighbors}

\author{Martin Aumüller}
{IT University of Copenhagen, Denmark}
{maau@itu.dk}
{}
{}

\author{Tobias Christiani}
{IT University of Copenhagen, Denmark}
{tobc@itu.dk}
{}
{}

\author{Rasmus Pagh}
{BARC and IT University of Copenhagen, Denmark}
{pagh@itu.dk}
{}
{}

\author{Michael Vesterli}
{IT University of Copenhagen, Denmark}
{miev@itu.dk}
{}
{}

\authorrunning{M. Aumüller, T. Christiani, R. Pagh, M. Vesterli}
\Copyright{Martin Aumüller, Tobias Christiani, Rasmus Pagh, Michael Vesterli}



\supplement{\url{https://github.com/puffinn/esa-paper}}

\conf{\relatedversion{arxiv link}}{}

\funding{The research leading to these results has received funding from the European Research Council under the European Union's 7th Framework Programme (FP7/2007-2013) / ERC grant agreement no.~614331. BARC, Basic Algorithms Research Copenhagen, is supported by the VILLUM Foundation grant~16582. }

\usepackage{spverbatim}

\begin{document}
\maketitle

\begin{abstract}
    We present PUFFINN, a parameterless LSH-based index for solving the $k$-nearest neighbor problem with probabilistic guarantees. By parameterless we mean that the user is only required to specify the amount of memory the index is supposed to use and the result quality that should be achieved. The index combines several heuristic ideas known in the literature. By small adaptions to the query algorithm, we make heuristics rigorous. We perform experiments on real-world and synthetic inputs to evaluate implementation choices and show that the implementation satisfies the quality guarantees while being competitive with other state-of-the-art approaches to nearest neighbor search. 
		We describe a novel synthetic data set that is difficult to solve for almost all existing nearest neighbor search approaches, and for which PUFFINN significantly outperform previous methods.
\end{abstract}

\section{Introduction}
\label{sec:intro}

\subsection{Our results} 

The $k$-nearest neighbor ($k$-NN) problem has been an object of intense research, both from theoretical and applied computer scientists.
There exist many implementations of $k$-NN data structures that perform very well in specific scenarios (see the related work section), but all implementations that we are aware of suffer from one or more of the following drawbacks:
\begin{itemize}
    \setlength\itemsep{0em}
\item Not scalable to large, high-dimensional data sets.
\item Not runtime-robust in the sense that query time may degrade to that of a linear search even for input distributions that are known to allow search in sublinear time.
\item Not recall-robust in the sense that there are input distributions that obtain low (less than 50\%) recall.
\item Performance bounds only hold for well-chosen values of certain parameters that depend on the data set as well as the query distribution.
\end{itemize}
PUFFINN combines several insights from recent theoretical research on $k$-NN data structures into a data structure that addresses these drawbacks.
Our contributions are as follows:
\begin{enumerate}
    \setlength\itemsep{0em}
    \item we present a parameterless and universal locality-sensitive hashing-based (LSH) implementation that 
        solves the $k$-NN problem with probabilistic guarantees (Section~\ref{sec:data:structure})
    \item we prove the correctness of an adaptive query mechanism building on top of the LSH forest data structure described by Bawa et al.~in~\cite{Bawa05} (Section~\ref{sec:data:structure})
    \item we describe an adaptive filtering approach to decrease the number of expensive distance computations (Section~\ref{sec:implementation})
    \item we propose a difficult dataset for the $1$-NN problem that exposes weaknesses in known heuristics (Section~\ref{sec:evaluation})
    \item we provide a detailed experimental study of our approach, evaluating design choices and relating it to the performance of other state-of-the-art approaches to $k$-NN search (Section~\ref{sec:evaluation})
\end{enumerate} 
Prior to this work, only subsets of these ideas have been implemented. 
Our main contribution on the theoretical side is that we make certain heuristics, such as the query algorithm described in~\cite{Bawa05}, rigorous. While some ideas of adaptive query algorithms have been discussed before~\cite{Dong08}, they made assumptions on the data and query distribution; our methods work for $k$-NN search in its full generality. On the practical side, we shed light on the empirical performance of theoretical ideas with regard to possible speed-ups of an LSH
implementation. 
Our final implementation is parameterless in the sense that it only requires the user to specify the space available for the data structure and the required quality guarantee.
Our implementation is competitive to state-of-the-art $k$-NN algorithms; in particular, it is as performant as previous LSH solutions that are not parameterless and do not provide guarantees on the quality of the result.


\subsection{Related work}

There exist many fundamentally different approaches to nearest neighbor search.
Popular techniques range from approximate tree-based methods, such as $kd$-trees~\cite{kdtree} and
$M$-trees~\cite{CiacciaPZ97} with an early stopping criterion~\cite{ZezulaSAR98}, to random projection trees~\cite{Dasgupta08}, 
to graph-based approaches~\cite{hnsw,Iwasaki18},
and finally hashing-based approaches, for example using LSH~\cite{IndykM98}. 
Each paradigm comes with performant implementations and we will introduce some of them in Section~\ref{sec:evaluation} when we are evaluating our implementation.
Since the focus of the present paper is designing a provably correct LSH implementation we 
will focus on existing methods in this realm. 
See the benchmarking paper by Aumüller et al.~\cite{AumullerBF17} for a more
detailed overview over other approaches and their performance on real-world datasets.

\subparagraph{Locality-sensitive hashing} 
LSH was introduced by Indyk and Motwani in~\cite{IndykM98}. We sketch the basic idea here. 
An LSH data structure consists of several independent repetitions of space partitions using LSH functions. 
An LSH function maps a data point to a hash code such that closer points are more likely to collide than far away points. 
In the LSH framwork, $K \geq 1$ locality-sensitive hash functions are concatenated to increase the gap between 
the collision probability of ``close'' points and ``far away'' points. 
For solving the $(c, r)$-near neighbor problem~\cite{IndykM98}, a certain concatenation length $K$ is fixed according to the number of points in the data set, the approximation factor $c$, and 
the strength of the hash family at hand. 
From the value $K$ and the hash family one can compute how many repetitions $L$ (using independent hash functions) have to be made to guarantee that a close point is found with constant probability.
The theoretical literature on LSH has mostly focused on solving the approximate near neighbor problem which can be be used to solve the approximate \emph{nearest} neighbor problem through a reduction~\cite{HarPeledIM12}.
In this paper we use LSH to solve the exact $k$-nearest neighbor problem with probabilistic guarantees.

\subparagraph{LSH implementations} 
Implementations of LSH evolved over the years with new advances in the theory of LSH. 
One of the first popular implementations, dubbed E2LSH~\cite{andoni20052},
was tailored for Euclidean space and included automatic parameter tuning of the $K$ parameter 
based on subsampling the dataset. 
From this $K$ parameter, other parameter such as the number of 
repetitions are derived. 
It solves the problem of reporting all points within a distance $r$ (specified during preprocessing) from the query and uses potentially large space 
on large datasets~\cite{andoni20052}.
The multiprobing approach introduced by Dong et al.~in~\cite{Dong08} allowed 
for implementations in which the space parameter can be fixed as in our approach. 
Their parameter tuning relies on the assumption that there is a certain distance distribution between queries and data points. 
LSHkit was the first implementation using this idea~\cite{lshkit}. 
A new LSH family for angular distance on the unit sphere motivated the development of the FALCONN library~\cite{andoni2015practical}. 
It contains highly optimized routines for efficient hash function computation, and supports multiprobing.

None of these approaches give guarantees on the query procedure if the query set is different with respect to distance distributions from the one seen during index building. 
Our data structure can be seen as a modified version of the LSH forest introduced by Bawa et al.~in~\cite{Bawa05}.
Instead of using a single hash length $K$, an individual repetition is a trie built on the hash codes of the data points.
Our query algorithm replaces the heuristic candidate collection of a predefined size with a rigorous termination criterion.


The cost of evaluating an LSH function differs widely. 
It ranges from $O(1)$ time for the bitsampling approach in Hamming space~\cite{IndykM98}, over $O(d)$ for random hyperplane hashing~\cite{Charikar02}, to $O(d^2)$ for cross-polytope  LSH~\cite{andoni2015practical}.
For the latter, the authors of~\cite{andoni2015practical} proposed a heuristic version that decreases the running time to $O(d \log d)$.
Another approach to reduce hashing time is to build a hashing oracle that returns the necessary $KL$ hash function values necessary to query the LSH data structure, but builds those from a smaller set of independent hash functions.
Christiani~\cite{Christiani17} described two approaches that reduce the amount of independent hash functions needed to produce these hash values from $KL$ to $K\sqrt{L}$ (using tensoring as in~\cite{andoni2006efficient}) or $O(\log^2 n)$ (using the pooling approach in~\cite{Dahlgaard2017FastSS}). 
While the E2LSH framework uses a variant of tensoring, we are not aware of an implementation using the pooling strategy.

Another idea that is currently missing in existing LSH implementations is the use of sketches, i.e., small representations of the original data points that allow to estimate the distance between two data points via their sketches. 
Christiani~\cite{Christiani17} describes how to use sketches when solving the near neighbor problem, but we are not aware of existing LSH-based implementations using this idea. 
We remark that sketching is a well-known technique and refer to the survey~\cite{Pagh19}.

\subparagraph{Auto Tuning Approaches} Apart from the approaches mentioned above, FLANN~\cite{flann} and the implementation of vantage point trees in nmslib~\cite{nmslib} are two non-LSH based nearest neighbor search that 
promise to tune the data structure to guarantee a certain quality criterion. This criterion is usually the recall of the query, i.e.,  the fraction of true nearest neighbor among the points returned by the implementation.

FLANN contains a collection of tree-based methods.  For auto-tuning, it takes a small sample of the data structure and builds indexes in a certain parameter space. 
It then queries the data structure with points from the data set and picks, among all the indexes that achieve at least the recall the user wishes for, the one with fastest query times.
The auto tuning employed by nmslib for the vantage point tree implementation follows the same principles and explores a certain parameter space based on a model of the data set to be indexed.
Both approaches require that the query and data set distribution are not too different. 
We will see in the experiments that both of the approaches do not satisfy the recall guarantees, even on real-world datasets.

\section{Preliminaries}
\label{sec:prelim}

\subsection{Problem Definition}
We assume a distance space $(\X, \text{dist})$ with distance measure $\text{dist}\colon\X \times \X \rightarrow \mathbb{R}_{\geq 0}$.
\begin{definition}
    Given a dataset $S \subseteq \X$ and an integer $k \geq 1$, the $(k, \delta)$ nearest neighbor problem ($(k, \delta)$-NN) is to build a data structure, 
    such that for every query $q \in \X$, 
    the query algorithm  returns a set of $k$ distinct points, each one being with probability at least $1 - \delta$ among the $k$ points in $S$ closest to $q$. 
\end{definition}
An algorithm solving the $(k, \delta)$-NN problem guarantees an expected recall of $(1 - \delta)k$, which is usually the quality measure in the context of nearest neighbor search algorithms. 

\subsection{Locality-Sensitive Hashing}
\begin{definition}[LSH Family~\cite{IndykM98, Charikar02}]
    \label{def:lsh}
    A \emph{locality-sensitive hash (LSH) family} $\HH$ is family of functions $h\colon X \to R$, such that for each pair $x, y \in X$ and a random $h\in\HH$, for arbitrary $q \in X$, whenever $\text{dist}(q,x)\le\text{dist}(q,y)$ we have $p(q, x) := \Pr[h(q)=h(x)]\ge\Pr[h(q)=h(y)]$.
\end{definition}
Traditionally, LSH families are used in the LSH framework to solve the $(c,r)$-near neighbor problem. \conf{}{For completeness, we provide a description of the standard framework in Appendix~\ref{app:lsh}.}

While the theory provided in this paper applies to every distance space that encompasses an
LSH family, we set our focus on solving the $k$-NN problem on the $d$-dimensional 
unit sphere under angular distance, which is equivalent to cosine similarity and inner product similarity on unit length vectors.
By using the Gaussian kernel approximation method of Rahimi and Recht~\cite{rahimi2008random} as described
by Christiani in~\cite{christiani2017framework}, 
our results extend to the whole Euclidean space. 

Random hyperplane (HP) LSH described by Charikar in~\cite{Charikar02} and Cross-Polytope (CP) LSH introduced 
by Terasawa and Tanaka~\cite{terasawa2007spherical} and analyzed by Andoni et al.~in~\cite{andoni2015practical} are two different LSH schemes under this distance measure.
A single HP LSH function produces a single bit. 
It works by choosing a random 
$d$-dimensional vector $a = (a_1, \ldots, a_d)$ where each $a_i$ is an independent standard normal 
random variable. 
The hash code of a point $x$ is $1$ if the inner product between $a$ and $x$ is at least 0, and 
0 otherwise.  A single CP LSH function applies a random rotation of $x$ on the unit sphere
and then maps it to the index of the closest vector out of the $2d$ signed standard basis vectors.
Technically, it can be thought of as choosing $d$ random hyperplanes $a_1, \ldots, a_d$ and mapping
$x$ to the index of the hyperplane with the largest absolute inner product, separating the two cases
that the inner product is negative or not.

\section{Data Structure}
\label{sec:data:structure}
This section describes the basic ideas of the data structure used in our implementation. 
Due to space reasons we only highlight the basic data structure and some of its properties.
The full description with all proofs can be found in \conf{the extended version of this paper}{Appendix~\ref{app:ds}}.
Note that the implementation has many differences to this clean version.
These differences are discussed in Section~\ref{sec:implementation}.

\subsection{Description}
In this section we will assume that we can perform distance computations and evaluate locality-sensitive hash functions in constant time.
Our data structure will be parameterized by integers $L, K \geq 1$ and will consist of a collection of $L$ LSH tries of max depth $K$.
This data structure is known as an LSH Forest~\cite{Bawa05}. 
Here we use a variant of the LSH tries with bounded depth and for completeness we include a brief description and statement of relevant properties.

\subparagraph*{LSH tries}
We index the LSH tries by $j = 1,\dots,L$. 
The $j$th LSH trie is built from the set of strings $\{ (h_{1,j}(x), \dots, h_{K,j}(x)) \mid x \in S \}$ where $h_{i,j} \sim \mathcal{H}$. 
The trie is constructed by recursively splitting the set of points $S$ on the next ($i$th) character until $|S| \leq i$ or $i = K + 1$ at which point we create a leaf node in the trie that stores references to the points in $S$.
Internal nodes store pointers to its children in a hash table where the keys are locality-sensitive hash values.

\subparagraph*{Query algorithm}
Let $S_{i,j}(q)$ denote the subset of points in $S$ that collide with $q$ when we consider the first $i$ hash values used in the construction of the $j$th trie. 
That is, $S_{i,j}(q) = \{ x \in S \mid h_{1,j}(q) = h_{1,j}(x) \land \dots \land h_{i,j}(q) = h_{i, j}(x) \}$.
For our query algorithm we wish to retrieve the points in each trie that collide with our query point in a bottom-up fashion, starting at depth $i = K$.
Define $\Pi_{i,j}(q) = S_{i,j}(q) \setminus S_{i+1, j}(q)$ where $S_{K+1,j}(q) = \emptyset$.
\begin{fact} \label{fact:lshtrie}
	LSH tries have expected construction time $O(nK)$ and use $O(n)$ words of space. 
	For $i \in \{0, 1, \dots, K\}$ we can retrieve a set $S^{'}_{i,j}(q) \supseteq S_{i,j}(q)$ with $|S^{'}_{i,j}(q)| \leq |S_{i,j}(q)| + i$ using time $O(|S_{i,j}(q)| + i)$.
	After having retrieved $S^{'}_{i,j}(q)$ we can retrieve $S^{'}_{i-1, j}(q)$ using additional time $O(|\Pi_{i-1,j}(q)|)$.
\end{fact}
The query algorithm is described in Algorithm~\ref{alg:adaptiveknn}.
During a query we search through the LSH Forest starting with the buckets $S_{i,j}(q)$ at depth $i = K$ and moving up one level once the $L$ LSH tries have been explored at the current level.
While searching we use a data structure PQ to keep track of the top-$k$ closest points seen so far. 
We stop the search once we have searched sufficiently many tries at a depth where our current $k$-nearest neighbor candidate would have been found with probability at least $1 - \delta$.
This stopping criterion ensures that we always search far enough to find the true $k$-nearest neighbor with probability at least $1 - \delta$. 

Insertions (Line~6) and retrieval of the $k$-largest distance (Lines~5 and 9) can be done in expected amortized time $O(1)$ by using an array of size $2k$ which is updated after each $k$ insertions.
For simplicity the algorithm is described as a double for-loop that iterates over the sets $\Pi_{i,j}(q)$ in a bottom-up fashion.
Using LSH tries Algorithm \ref{alg:adaptiveknn} would be implemented by using a straight-forward bottom-up traversal of the tries with properties described in Fact~\ref{fact:lshtrie}.

We proceed by proving that Algorithm \ref{alg:adaptiveknn} solves the $(k, \delta)$-NN problem as well as providing running time bounds. 
The idea behind an adaptive $k$NN algorithm with guarantees following the approach of Algorithm~\ref{alg:adaptiveknn} can be attributed to Dong et al.~\cite{Dong08}.
Christiani et al.~\cite{christiani2018confirmation} show how a different stopping criteria gives a self-tuning algorithm in the regime where $\delta = 1/n$. 
To the best of our knowledge both the following proof of correctness (although simple) and the running time bound for Algorithm \ref{alg:adaptiveknn} is new.
\begin{algorithm}[t]
\SetKw{or}{ or }
\SetKw{and}{ and }
\DontPrintSemicolon
PQ $\gets$ empty priority queue of $(\text{point}, \text{dist})$ of unique points \;
\For{$i \gets K, K - 1, \dots, 0$}{
	\For{$j \gets 1, 2, \dots, L$}{
        \For{$x \in \Pi_{i,j}(q)$}{
            \tcc{We abbreviate $x_k' \gets \mathrm{PQ.max()}$ for ease of notation}
        \If{$\mathrm{dist}(q, x'_k) \geq \mathrm{dist}(q, x)$ }{PQ.insert$(x, \dist(q, x))$\tcp*{Remove largest entry if PQ contains more than $k$ elements.}}}
        \If{$i = 0 \or (\mathrm{PQ.size()} == k   \and j \geq \ln(1/\delta)/p(q, x'_k)^i)$}{\Return PQ 
	}
}
}
\caption{\textsc{adaptive-kNN}$(q, k, \delta)$} \label{alg:adaptiveknn}
\end{algorithm}
\begin{lemma}
	\textsc{adaptive-kNN}$(q, k, \delta)$ returns a set of $k$ points, each one being with probability at least $1-\delta$ among the closest $k$ points to $q$.
\end{lemma}
\begin{proof}
    As introduced in Definition~\ref{def:lsh}, we use the short notation $p(q, x) := \Pr[h(q) = h(x)]$ under the random LSH hash function choice $h$. Let $x_1, \ldots, x_k$ be the $k$ closest neighbors of $q$ in $S$. First, observe that at any stage the algorithm maintains the invariant ``$\mathrm{PQ.size()} < k \textbf{ or } p(q, \mathrm{x_k'}) \leq p(q, x_k)$''. This is true because at any point, the $k$-th closest point $x'_k$ identified by the algorithm satisfies $\dist(q, x'_k) \geq
    \dist(q, x_k)$. Together with the montonicity of the collision probability of LSH, cf.~Definition~\ref{def:lsh}, the invariant holds. Thus, the algorithm cannot terminate until $j'$ tries have been searched at level $i'$ where either $j' \geq \ln(1/\delta)/p(q, x_k)^{i'}$ or $i' = 0$.
	In the first case the probability of not finding a $k$NN of $q$ is at most $(1 - p(q, x_k)^{i'})^{j'} \leq \delta$. 
	In the case of $i' = 0$ the query algorithm degrades to a linear scan and we are guaranteed to report the true $k$NNs. 
\end{proof}
Next, we connect the expected running time of Algorithm~\ref{alg:adaptiveknn} to the optimal expected running time of an algorithm that knows optimal parameter choices for $i$ and $j$. 
We will use $OPT(L, K, k, \delta)$ to denote the optimal expected query time that can be achieved with the natural algorithm that solves $k$NN queries on the LSH Forest by searching $j$ tries at depth $i$ where $i$ and $j$ are chosen to minimize the query time. 
In our expression for the expected the query time we use a unit cost model that counts hash function evaluations and distance computations.
To ensure that each point in the $k$NN set is reported with probability at least $1-\delta$ we search $j = \ln(1/\delta)/p(q, x_k)^i$ tries. 
The expected cost of searching one LSH trie at depth $i$ is $i + \sum_{x \in P} p(q, x)^i$.
\begin{equation*}
OPT(L, K, k, \delta) = \min\left\{ \frac{\ln(1/\delta)}{p(q, x_k)^i}(i + \sum_{x \in P} p(q, x)^i )  \mid 0 \leq i \leq K, \frac{\ln(1/\delta)}{p(q, x_k)^i} \geq L \right\}. 
\end{equation*}
\conf{
We obtain the following lemma with a proof provided in the extended version of this paper. 
}
{
We obtain the following lemma with a proof provided in Appendix~\ref{app:ds}.
}
\begin{lemma} \label{lem:knn_time}
	Let $0 < \delta \leq 1/2$ then with probability $1 - \delta$ we have that \textsc{adaptive-kNN}$(q, k, \delta)$ terminates in expected time $O(OPT(L, K, k, \delta/k) + L(K + k))$. 
\end{lemma}

\subsection{Reducing the Number of LSH Evaluations}
In~\cite{Christiani17}, Christiani provides a uniform framework encompassing previous ad-hoc solutions~\cite{Dahlgaard2017FastSS,andoni2006efficient} to reduce the amount of hash function evaluations when solving the approximate near neighbor problem. \conf{In the extended version of this paper}{In Appendix~\ref{app:ds}} we describe how these techniques can be applied when solving the $k$-NN problem. In particular, each method requires an adaption of Algorithm~\ref{alg:adaptiveknn} and comes with their own stopping criterion. We provide a succinct description of the methods next. 

\subparagraph*{Tensoring} 
Assume that $K$ is an even integer and $L$ is an even power of two.
Form two collections of $\sqrt{L}$ tuples of $K/2$ LSH functions.
Each trie in the LSH forest is now indexed by $j_1, j_2 \in \{1, \dots, \sqrt{L}\}$.
The $K$ LSH functions used in the $(j_1, j_2)$st trie are taken by interleaving the $K/2$ functions in the $j_1$st tuple of the first collection and the $j_2$nd tuple of the second collection.
This allows us to construct $L$ LSH tries of max depth $K$ using only $\sqrt{L}K$ independent functions.

\subparagraph*{Pooling} 
Form a pool of $m$ independent LSH functions that will be shared among LSH tries.
For each LSH trie in the LSH Forest we independently sample a random subset of $K$ LSH functions from the pool that we use in place of fully random LSH functions.
The LSH pool can be viewed as a randomized construction of a smaller LSH family from our original LSH family.
As $m$ increases LSH functions sampled without replacement from the pool will work almost as well as independent samples from the LSH family.

\subsection{Sketching for faster distance computations}
Locality-sensitive hashing can be used to produce $1$-bit sketches for efficient similarity estimation~\cite{Charikar02, li2011theory}.
The idea is to use a random hash function to hash the output of a locality-sensitive hash function to a single bit, and then packing $w$ such bits into a $w$-bit machine word.
We can use word parallelism (alternatively table lookups) to count the number of collisions between $w = \Theta(\log n)$ such sketches in $O(1)$ time, allowing us to efficiently estimate the similarity between points in our original space.
Depending on the LSH scheme used and the distribution of distances between the query point and the data, using $1$-bit sketches can replace many of the expensive distance computations performed by the query algorithm with cheaper distance estimations through sketching. 
See~\cite{satuluri2012bayesian,Christiani17} for more details on sketching in the context of the ANN problem.

\section{Implementation Overview}
\label{sec:implementation}

\subsection{Overall structure}

\begin{figure}
    \centering
    \includegraphics[width=\textwidth]{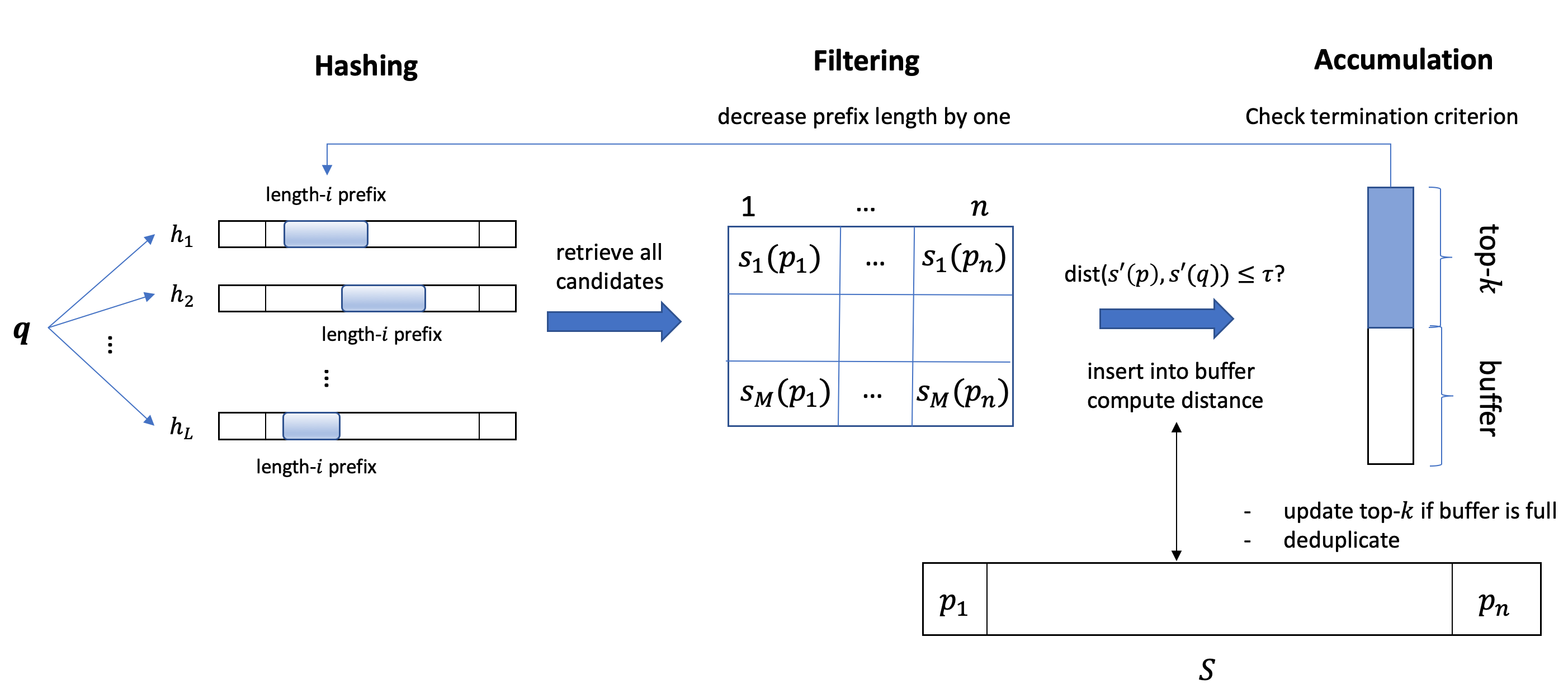}
    \caption{Overall structure of our implementation.}
    \label{fig:structure}
\end{figure}

Figure~\ref{fig:structure} presents an overview over our data structure. Deviating from the pointer-based trie data structure described in Section~\ref{sec:data:structure}, we use an array of indices sorted by hash code, which improves both cache- and space-efficiency.
Additionally, a sketching-based filtering approach is used to 
reduce the number of distance computations carried out during a query. In the following, we make the implementation precise.

A query is answered in rounds, starting at the maximum considered prefix length. 
A single round works as follows:
Repetitions are inspected one after the other. In each repetition, 
all (new) candidates sharing the hash prefix with the query point are retrieved. 
For each such candidate point, a sketch is chosen and checked against the corresponding sketch of the query point. 
Let $\tau$ denote a threshold value that we will discuss how to set later.
If the Hamming distance between candidate and query sketch is less than $\tau$,
the data point passes the filter, the distance computation is carried out, and the point with its distance to the query is inserted into the accumulator buffer.
Once this buffer is full we discard all points not belonging to
the top-$k$. 
While this can be done in time $O(k)$, we found that an implementation based on sorting was faster for the values of $k$ we considered.
The termination criterion is checked after all repetitions are inspected.  
If the criterion is satisfied, the algorithm returns the indices of the top-$k$ points in the accumulator, otherwise the prefix length is decreased by one and a new round starts.

The following subsections make this process more detailed and discuss engineering choices
in the case of angular distance on unit length vectors. 

\subsection{Engineering choices}

\subparagraph{Vector storage}
We normalize the query vector and all data vectors, which means that all dot products will be between -1 and 1. 
This allows storing vectors in a fixed point format represented using 16-bit integers. 
Such a representation enables AVX2-enabled instructions that allow 16 multiplications at once. 
To use the AVX instructions, all vectors are stored in a 1-dimensional array with padding to ensure 256-bit alignment.

\subparagraph{Retrieving candidate points}
The $j$th LSH repetition is represented as a sorted array of tuples of the form ($h_{\leq i, j}(p), \text{pos}(p)$),
where $\text{pos}(p)$ is the index of $p \in S$ in the original dataset, and $h_{\leq i, j}(p)$ is the hash code
of $p$ under hash functions $h_{1, j} \circ \cdots \circ h_{i, j}$. We view the hash code as a bitstring.
Before retrieving candidates, we first find the tuple with the longest common prefix with the hash of the query vector. 
This is achieved using binary search, where we tabulate the lexicographic position of each 13-bit prefix to speed up the search. 

Each time more candidates are requested, all tuples whose hash code has a common prefix of length at least $i$ are considered.
Each iteration decrements $i$ in order to increase the number of considered vectors. 
Since the vectors are stored in sorted order, all tuples with a common prefix of length $i$ are stored adjacent to the tuples with common prefixes of length $i+1$. 
Furthermore, they are all stored on the same side, depending on whether the removed bit was a 0 or 1. 
This means that the range of considered vectors can be updated efficiently. 

Every access in the array is done in a segment of size $B = 12$, regardless of whether the prefix matches or not. This costs almost no time, because the random memory access is the expensive part, and only improves quality. A discussion of suitable values of $B$ is provided in \conf{the extended version of this paper}{Appendix~\ref{app:segment:size}}.

\subparagraph{Filtering candidate points}
The filtering step is an additional measure to reduce the number of distance computations. Fix a point $p \in S$. 
During the index building phase, we store $M$ 64-bit sketches $s_1(p), \ldots, s_M(p)$ obtained via HP LSH.
If $p$ is retrieved as a candidate, retrieve a sketch $s'(p)$ using a pseudorandom transformation of the repetition number $j$.
Next, compute the Hamming distance between $s'(p)$ and $s'(q)$. 
If the distance is at most $\tau \in \{0,\ldots, b\}$, $p$ passes the filter and is inserted into the accumulator. 
A challenge in the context of $k$-NN queries is that the algorithm does not know the distance to the $k$-th nearest neighbor. This means that the threshold has to be adapted according to the points inspected so far. We set
threshold $\tau$ dynamically according to the probability that a vector with Hamming distance $\tau$ or less has a dot product larger than the smallest dot product in the current top-$k$. 

\subparagraph{Computing distances}
The accumulator takes care of the candidate points that pass the filter step. 
It de-duplicates the candidate list and keeps track of the top-$k$ points found so far. 
The accumulator consists of a buffer of size $2k$, which contains the current top-$k$ indices, along with their dot products, and a buffer of size $k$, which contains points that passed the filter along with their dot products.
Once this buffer is full, the top-$k$ list is updated. 

\subsection{Locality-sensitive Hash Functions}

\subparagraph{Supported LSH Functions}
The supported hash functions relevant for the paper are HP LSH~\cite{Charikar02} and CP LSH~\cite{andoni2015practical}.\footnote{PUFFINN is generic to the LSH family, but some engineering choices are different for other similarity measures, such as set similarity. At the moment, PUFFINN also support ($b$-bit) MinHash~\cite{bro97b,li2011theory}.} 
For the latter, the implementation encompasses both the exact version and the pseudorandom version with three applications of the fast Hadamard transform, see~\cite{andoni2015practical} for more details.
We always regard hash functions as producing an $\ell$-bit string as its output. For HP LSH, we have $\ell = 1$, for CP LSH, we have $\ell = \lceil \log 2d\rceil$.
In case the algorithm did not terminate after exploring all $L$ repetitions, decreasing the prefix length by one always means that we disregard the last bit of the hash.
This is to avoid a sudden increase in the number of collisions in the case of CP LSH. This is theoretically sound since the termination criteria from Section~\ref{sec:data:structure} only need a lower bound on the collision probability at a certain prefix-length, which can be estimated for individually bit lengths of CP LSH. 

\subparagraph{Estimating Collision Probabilities}
Recall from Section~\ref{sec:data:structure} that evaluating the collision probability of two points at a certain distance is a key ingredient in the query algorithm. 
While such a formula is easy to derive for HP LSH, we only know of the asymptotic behavior of collision probabilities for CP LSH~\cite{andoni2015practical}.  
To overcome this obstacle, we find a Monte Carlo estimate on the collision probability of unit vectors with inner product $\alpha, -1 \leq \alpha \leq 1$, by enumerating different values of $\alpha$ in a window of size .05. 
For a fixed distance, we consider two points $x = (1, 0, \ldots, 0)$ and $y = (\alpha, \sqrt{1 - \alpha^2}, 0, \ldots, 0)$,\footnote{By spherical symmetry, the collision probabilities are the same for all pairs of points with inner product $\alpha$.} draw 1\,000 random CP hash functions, count the number of collisions, and tabulate the estimate. 
As mentioned above, we always consider bit strings, so the probability estimation for CP LSH is made for all bit lengths up to $\ell = \lceil \log 2d\rceil$. 

In the query procedure, we round the distance down to the closest distance value for which we have tabulated an estimate and use that to bound the collision probability. 
The evaluation in the next section will show that this yields a negligible loss in quality compared to an exact variant using HP LSH.

\section{Experimental Evaluation}
\label{sec:evaluation}

\subparagraph{Implementation and Experimental Setup} PUFFINN is implemented in C++ and comes with a wrapper to the Python language. 
Experiments were run on 2x 14-core Intel Xeon E5-2690v4 (2.60GHz) with 512GB of RAM using 
Ubuntu 16.10 with kernel 4.4.0. It is compiled using \texttt{g++} with the compiler flags 
\texttt{-std=c++14 -Wall -Wextra -Wno-noexcept-type -march=native -O3 -g -fopenmp}. Index building was multi-threaded, queries 
where answered sequentially in a single thread. The experiments were conducted in the \texttt{ann-benchmarks} framework from~\cite{AumullerBF17}. The code, raw experimental results, and the Jupyter notebook used for the evaluation are available at \url{https://github.com/puffinn/esa-paper}.

\subparagraph{Quality and Performance Metrics} As quality metric we measure the
individual recall of each query, i.e., the fraction of points reported by the
implementation that are among the true $k$-NN. As performance metric, we record 
individual query times. We usually report on the \emph{throughput}, i.e., the average number of queries that can
be answered in one second. In plots, the throughput is dubbed QPS for \emph{queries per second}.

\subparagraph{Parameter Choices} PUFFINN has two parameters: the space a user is willing 
to allocate for the index, and the expected recall that should be achieved. We run PUFFINN with expected 
recall values in the set $\{.1, .2, .5, .7, .9, .95\}$. As space parameters, we use the 
doubling range 512 MB to 32 GB. We always retrieve the ten nearest neighbors.

\subparagraph{Objectives of the Experiments} 
Our experiments are tailored to answer the following high-level questions (\emph{HL-Q}): 



\begin{enumerate}[leftmargin=1.5cm]
    \setlength\itemsep{0em}
    \item[(HL-Q1)] Given the choices of sketching and hash function evaluation methods described in the 
        previous two sections,
        how do they compare to each other w.r.t. empirical performance (Sections~\ref{sec:eval:filtering}--\ref{sec:eval:hashing})? 
    \item[(HL-Q2)] Can a parameterless method compete with implementations that have parameters tuned to the data and query workload (Section~\ref{sec:eval:comparison})?  
\end{enumerate}

\noindent To answer these questions, we will consider the following implementation-level questions (\emph{IL-Q}):

\begin{enumerate}[leftmargin=1.5cm]
    \setlength\itemsep{0em}
    \item[(IL-Q1)] What is the influence of the filtering approach to the quality/QPS (Section~\ref{sec:eval:filtering})?
    \item[(IL-Q2)] What is the influence of the update threshold $\tau$ to quality and QPS 
        (Section~\ref{sec:eval:filtering})?  
    \item[(IL-Q3)] How does the space parameter influence the QPS (Section~\ref{sec:eval:size})? 
    \item[(IL-Q4)] How does the hash function and evaluation strategy influence performance (Section~\ref{sec:eval:hashing})?  
\end{enumerate}

\subparagraph{Real-World Datasets}
\begin{table*}[t]
  \small
  \centering
  \begin{tabular}{l r r}
    \toprule
    \textbf{Dataset} & \textbf{Data Points/Query Points} & \textbf{Dimensions} \\
    \midrule
    \textsf{GLOVE}~\cite{pennington2014glove} & 1\,183\,514/10\,000 & 100 \\
    \textsf{GLOVE-2M}~\cite{pennington2014glove}& 2\,196\,018/10\,000& 300 \\
    \textsf{GNEWS-3M}~\cite{word2vec} & 3\,000\,000/10\,000   & 300 \\
    \textsf{SYNTHETIC} & 1\,000\,000/\phantom{0}1\,000 & 300 \\
    \bottomrule
  \end{tabular}
  \caption{Datasets under consideration.} 
\label{tab:datasets}
\end{table*}

Table~\ref{tab:datasets} gives an overview of the datasets used in the experiments. 
Cosine similarity is usually used in the context of word embeddings, so we use three 
real-world datasets that originate from two different word embedding algorithms. Unless stated otherwise, all experiments are carried out on \textsf{Glove-1M}. 

\subparagraph{Synthetic Dataset}

We describe a synthetic data set and query distribution that, as we will see, is challenging for many heuristic nearest neighbor implementations.
For a fixed $d \geq 1$, we construct a dataset over $\mathbb{R}^{3d}$ as follows. 
For each $i \in \{1, \ldots, n - 1\}$, let $y_i$ and $z_i$ be two $d$-dimensional vectors of expected length $\sqrt{1/2}$ where each coordinate is sampled independently from $\mathcal{N}(0, 1/{2d})$.  
Define $x_i = (0^d, y_i, z_i)$. Let $v$ and $w$ be two more random vectors 
of expected length $\sqrt{1/2}$.
Finally, set $x_n = (v, w, 0^d)$.
We define $m$ query vectors as follows:
For each $i \in \{1, \ldots, m\}$, let $q_i = (v, 0^d, r_i)$, 
where $r_i$ is a random vector of length $\sqrt{1/2}$.

This construction has the property that $x_n$ is the nearest neighbor of every query.
Furthermore, all data points have unit length in expectation.
The distance from each $q_i$ to $x_n$ is expected to be 1 (or equivalently $\E[\ip{q_i}{x_n}] = \frac{1}{2}$), whereas the distance to all other points is around $\sqrt{2}$ (i.e., $\E[\ip{q_i}{x_j}] = 0$).  
In the experiment, we choose $n = 1\,000\,000$ and $d = 100$.

\subparagraph{Other approaches}
We compare PUFFINN to the following implementations:
        FALCONN, a state-of-the-art LSH implementation using the theory developed in~\cite{andoni2015practical};
        ONNG, a recent graph-based approach described in~\cite{Iwasaki18};
        ANNOY, the best-performing implementation of a random-projection forest~\cite{annoy};
        IVF, a $k$-means clustering based approach~\cite{faiss};
        FLANN, a collection of different approaches with tuning of recall value~\cite{flann};
        VantagePointTree~\cite{Yianilos93} as implemented in NMSlib~\cite{nmslib} with recall guarantees.

These approaches stood out in the benchmarking paper from Aumüller et al.~\cite{AumullerBF17} as performing best on many datasets. 
We use the same parameter space as in~\cite{AumullerBF17} to test the performance of the different implementations. For each implementation, we report the best results achieved via a grid search over the (usually large) parameter space. 
We refer to that paper or the original papers for more details on these approaches.
Except from FLANN and VantagePointTree, no other implementation allows to specify a guarantee on recall.

\subsection{Filtering Approach}
\label{sec:eval:filtering}

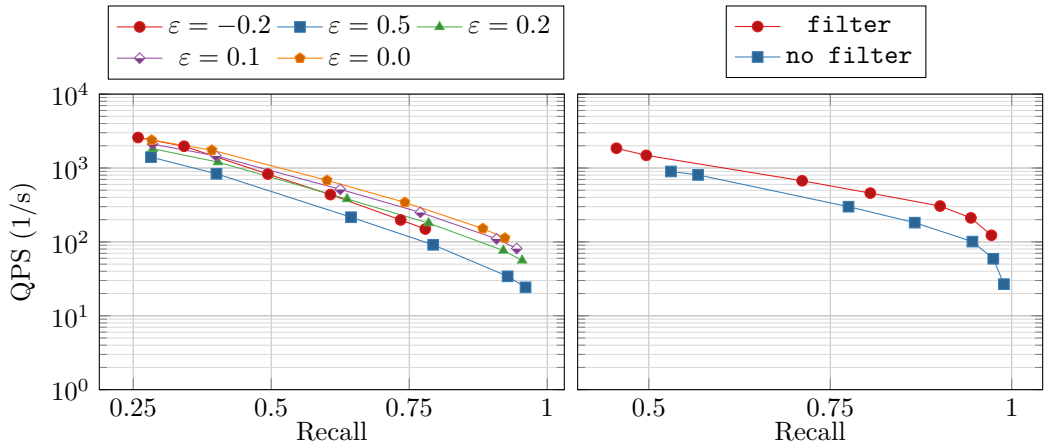
\begin{figure}
    \begin{tikzpicture}[every mark/.append style={mark size=2pt}]
        \begin{groupplot}[group style = {group size = 2 by 1, horizontal sep = .5em, group name = robustness}, height=5.5cm, width=.55\textwidth,grid = both, grid style={line width=.1pt, draw=gray!30}, major grid style={line width=.2pt,draw=gray!50}, ylabel style = {yshift=-1.5ex}, xlabel style = {yshift=1.5ex}, ymin = 1, ymax = 10000, xtick = {0, 0.25, 0.5, 0.75, 1}, max space between ticks=20]
            \nextgroupplot[
                xlabel = {Recall},
                ylabel = {QPS (1/s)},
                ymode=log,
                legend style = {opacity = 1.0, text opacity = 1, at = {(1, 1.3)}},
                legend columns = 3,
                grid = both, grid style={line width=.1pt, draw=gray!30},
major grid style={line width=.2pt,draw=gray!50},
            ]
            \addplot coordinates { (0.258730, 2587.712787) (0.342180, 1970.795206) (0.493920, 827.650660) (0.607140, 436.821394) (0.734470, 198.127343) (0.779110, 149.310894) };
\addlegendentry{$\varepsilon=-0.2$};
\addplot coordinates { (0.282170, 1404.429266) (0.400740, 834.940250) (0.644530, 215.779705) (0.793400, 91.314531) (0.928280, 34.208335) (0.960760, 24.272033) };
\addlegendentry{$\varepsilon=0.5$};
\addplot coordinates { (0.285410, 1822.651786) (0.403560, 1202.653164) (0.637440, 380.884954) (0.785120, 179.788312) (0.920550, 76.541051) (0.954590, 56.166717) };
\addlegendentry{$\varepsilon=0.2$};
\addplot coordinates { (0.285120, 2105.482561) (0.399780, 1461.745181) (0.625080, 514.225790) (0.770220, 250.630654) (0.908220, 109.930540) (0.944770, 81.413828) };
\addlegendentry{$\varepsilon=0.1$};
\addplot coordinates { (0.283340, 2389.254889) (0.392390, 1745.706468) (0.600910, 681.987872) (0.741980, 342.419761) (0.883310, 151.578721) (0.923270, 112.862351) };
\addlegendentry{$\varepsilon=0.0$};
            \nextgroupplot[
                xlabel = {Recall},
                ymode=log,
                yticklabels = {,,},
                legend style = {opacity = 1.0, text opacity = 1, at = {(.75, 1.3)}},
                grid = both, grid style={line width=.1pt, draw=gray!30},
major grid style={line width=.2pt,draw=gray!50},
            ]

\addplot coordinates { (0.455340, 1851.714399) (0.496530, 1483.303657) (0.711220, 671.488070) (0.805140, 457.223737) (0.901290, 305.465230) (0.943720, 211.285127) (0.971990, 122.818919) };
\addlegendentry{filter};
\addplot coordinates { (0.530490, 898.193404) (0.567850, 805.809243) (0.775120, 299.487249) (0.866350, 182.199221) (0.945830, 100.939947) (0.974600, 59.115333) (0.989220, 26.795525) };
\addlegendentry{no filter};
\end{groupplot}\end{tikzpicture}
\caption{Left: Influence of setting the threshold of sketches to $1 + \varepsilon$ times the expected 
difference at the distance of the $k$-th closest point found so far. Expected recall values: 0.1, 0.2, 0.5, 0.7, 0.9, 0.95, space: 1GB. Right: Difference between filtering/no filtering, space: 4GB.}
\label{plot:sketch:eps}
\end{figure}

We evaluate the filtering approach in two directions. 
First, we report on the quality-performance trade-off of different update strategies. 
Second, we benchmark the architecture against a ``no filtering'' approach.
Experiments were done using a collection of $32$ sketches using $64$ bits each. 

Figure~\ref{plot:sketch:eps} (top) reports on the influence of setting the passing threshold of the filtering step dynamically to a fraction of $\tau = 1 + \varepsilon$ of the expected difference\footnote{The expected difference is just 64 times the probability that two points at distance $r_k$ collide under a random hyperplane.} for a point at the distance of the current $k$-th nearest neighbor. 
A $\varepsilon$-value of $0.0, 0.1$, and $0.2$ give good results in this empirical setting.
Above $0.1$ there is almost no gain in quality but a huge drop in QPS. 
Setting the threshold below the expectation results in a large loss in quality. 
For the remainder of the experiments, we set the threshold to the expectation, i.e., we use $\delta = 0$. 

Figure~\ref{plot:sketch:eps} (bottom) allows us to see the difference between using resp. not using the filtering approach. 
For low recall values, the filtering approach increases the QPS by a factor of roughly 1.5. 
For example, the filtering approach can answer around 1\,400 QPS at recall .5, without filtering this number drops to 900. 
At high recall, the difference is more pronounced. 
A recall of 97\% is achieved with 122 QPS using filtering and the same recall is achieved at around 50 QPS without filtering.
We can see a clear difference in the achieved quality between the two variants. 
The sketching approach usually decreases the recall for a fixed expected recall by .08 (for low recall) to .03 (for high recall). 
However, the recall is still above the guarantee in both cases. 


\subsection{Influence of the Index Size}
\label{sec:eval:size}

\begin{figure}
    \begin{tikzpicture}[every mark/.append style={mark size=2pt}]
            \begin{axis}[
                xlabel = {Recall},
                ylabel = {QPS (1/s)},
                width = 1\textwidth,
                height = 5.5cm,
                ymode=log,
                legend style = {opacity = 1.0, text opacity = 1, at = {(0.725, 0.25)}},
                legend columns = 5, 
                grid = both, grid style={line width=.1pt, draw=gray!30},
major grid style={line width=.2pt,draw=gray!50},
            ]

\addplot coordinates { (0.325250, 1131.514702) (0.411850, 565.772778) (0.617810, 218.440290) (0.727810, 127.353683) (0.840960, 26.242067) (0.900190, 12.444223) (0.965460, 7.079529) };
\addlegendentry{512 MB};
\addplot coordinates { (0.338850, 2611.143601) (0.426910, 1788.637379) (0.620560, 661.041328) (0.728170, 354.297099) (0.874050, 125.625582) (0.914460, 92.609257) (0.967530, 51.338165) };
\addlegendentry{1GB};
\addplot coordinates { (0.367200, 2515.122205) (0.470590, 1671.278202) (0.661030, 762.693220) (0.776500, 474.166746) (0.891390, 219.366445) (0.922510, 151.620309) (0.967630, 74.147491) };
\addlegendentry{2 GB};
\addplot coordinates { (0.455340, 1851.714399) (0.496530, 1483.303657) (0.711220, 671.488070) (0.805140, 457.223737) (0.901290, 305.465230) (0.943720, 211.285127) (0.971990, 122.818919) };
\addlegendentry{4 GB};
\addplot coordinates { (0.611980, 973.464294) (0.619240, 992.908636) (0.679530, 741.272429) (0.801660, 439.109809) (0.914440, 244.890016) (0.952720, 194.159743) (0.978370, 135.779841) };
\addlegendentry{8 GB};
\end{axis}\end{tikzpicture}
\caption{Influence of index size to quality-performance trade-off.}
\label{plot:index:size}
\end{figure}
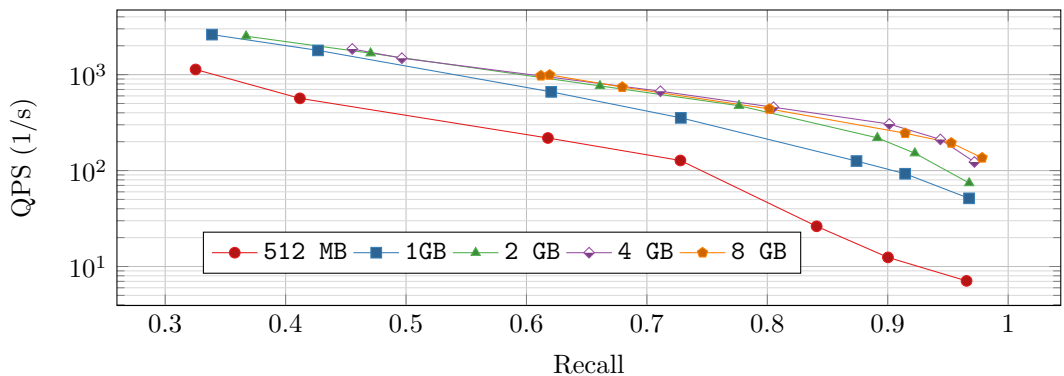

We turn our focus to the third implementation-level question: How does the index size influence the performance? We note that the index size includes the whole data structure, including storage of the original dataset and the hash functions.
Figure~\ref{plot:index:size} reports on the quality-performance trade-off achieved by the implementation for different index sizes. 
We observe that larger index sizes provide better performance, but the influence is diminishing at more than 4 GB. 
A small index yields a small number of repetitions, which means that the algorithm has to explore many levels in the data structure. 
For a recall of .9, increasing the amount of space from 512 MB to 2 GB increases the QPS from 12 to roughly 200. 
Doubling the allotted space to 4GB results in around 300 QPS, which is roughly the same for 8 GB as well. 
We can see that the achieved recall is above the set guarantee threshold for all tested index sizes. (Each data point
corresponds to a recall value from the set of tested recall values.)

\subsection{Choice of Hash Function and Evaluation Strategy}
\label{sec:eval:hashing}

We start by evaluating different hash evaluation strategies. 
We implemented the following three different evaluation strategies in PUFFINN: independent, 
tensor, and pool. For the pooling strategy, we set up a pool containing 3072 bits. 
In the following, we fix the hash function used to be CP LSH using fast Hadamard transform.
Figure~\ref{plot:eval:strategies} shows a comparison between the three evaluation strategies for different index sizes.
As we can see, tensoring is never better than the pooling strategy. 
Furthermore, independent gives better performance for a fixed quality for large index sizes, but takes
more time for initializing the hash values at low recall.
We note that when using the exact CP LSH, there is a huge difference between independent and pooling, in particular for large index sizes. For example, CP using independent hash functions achieves not more than 80 QPS for the 8 GB index in the right plot in Figure~\ref{plot:eval:strategies}.
For all of these reasons, we fix the implementation to use the pooling strategy.

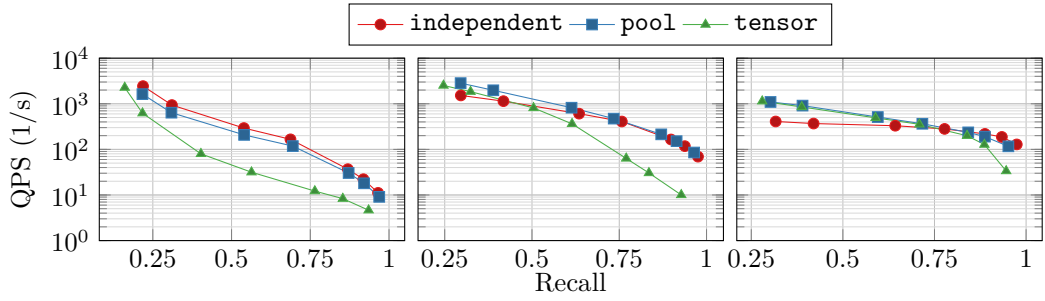
\begin{figure}
    \begin{tikzpicture}[every mark/.append style={mark size=2pt}]
        \begin{groupplot}[group style = {group size = 3 by 1, horizontal sep = .5em, group name = robustness}, height=4cm, width=.4\textwidth,grid = both, grid style={line width=.1pt, draw=gray!30}, major grid style={line width=.2pt,draw=gray!50}, ylabel style = {yshift=-1.5ex}, xlabel style = {yshift=1.5ex}, ymin = 1, ymax = 10000, xtick = {0, 0.25, 0.5, 0.75, 1}, max space between ticks=20]
            \nextgroupplot[
                xlabel = {},
                ylabel = {QPS (1/s)},
                ymode=log,
                legend style = {opacity = 1.0, text opacity = 1, at = {(2.4, 1.3)}},
                legend columns = 3,
                grid = both, grid style={line width=.1pt, draw=gray!30},
major grid style={line width=.2pt,draw=gray!50},
            ]

\addplot coordinates { (0.218860, 2439.874901) (0.310630, 937.482289) (0.539220, 293.450016) (0.687090, 166.103970) (0.869310, 36.824938) (0.918720, 22.182805) (0.965280, 11.125753) };
\addlegendentry{independent};
\addplot coordinates { (0.217050, 1627.530098) (0.309070, 637.729792) (0.539980, 207.844760) (0.694440, 118.066237) (0.871780, 30.319863) (0.920850, 18.030915) (0.968930, 9.099753) };
\addlegendentry{pool};
\addplot coordinates { (0.160860, 2269.313442) (0.217340, 622.028026) (0.402780, 79.900321) (0.563330, 31.585476) (0.764390, 12.184354) (0.853330, 8.317483) (0.935180, 4.639254) };
\addlegendentry{tensor};
            \nextgroupplot[
                xlabel = {Recall},
                ylabel = {},
                ymode=log,
                yticklabels = {,,},
                legend style = {opacity = 1.0, text opacity = 1},
                grid = both, grid style={line width=.1pt, draw=gray!30},
major grid style={line width=.2pt,draw=gray!50},
            ]

\addplot coordinates { (0.296140, 1523.817654) (0.418220, 1141.607159) (0.634920, 607.453475) (0.757270, 408.849308) (0.898060, 164.988520) (0.937830, 117.763050) (0.975500, 69.286258) };
\addplot coordinates { (0.295600, 2837.869323) (0.389160, 1977.281027) (0.613260, 820.133249) (0.734160, 472.841963) (0.870270, 213.582206) (0.914250, 151.447258) (0.964430, 85.402692) };
\addplot coordinates { (0.246490, 2521.450535) (0.324000, 1846.328935) (0.504450, 810.567865) (0.614720, 363.442077) (0.769920, 63.390998) (0.834950, 30.404024) (0.927590, 10.025918) };
            \nextgroupplot[
                xlabel = {},
                ylabel = {},
                yticklabels = {,,},
                ymode=log,
                legend style = {opacity = 1.0, text opacity = 1},
                grid = both, grid style={line width=.1pt, draw=gray!30},
major grid style={line width=.2pt,draw=gray!50},
            ]

\addplot coordinates { (0.316590, 409.186979) (0.420020, 368.664506) (0.642970, 330.060063) (0.777570, 281.204366) (0.887680, 218.554135) (0.933850, 187.681813) (0.974490, 128.432771) };
\addplot coordinates { (0.302500, 1095.998051) (0.389860, 912.374386) (0.595260, 511.679217) (0.716930, 364.523481) (0.842080, 236.199456) (0.887480, 190.263860) (0.951630, 116.554429) };
\addplot coordinates { (0.279900, 1137.428960) (0.387790, 852.257303) (0.589160, 491.674136) (0.708890, 354.676337) (0.838900, 199.400008) (0.886220, 127.602541) (0.945710, 33.554830) };
\end{groupplot}
\end{tikzpicture}
\caption{Comparison of hash evaluation strategies. Left: 512 MB, center: 2 GB, right: 8 GB.}
\label{plot:eval:strategies}
\end{figure}

We turn our focus to the choice of hash function. 
Figure~\ref{plot:eval:functions} gives an overview of three different index sizes using the pooling strategy for hash function evaluation.
On the smallest index, HP LSH works well, in particular for high recall. 
If there is space for more repetitions, CP LSH becomes the method of choice, and 
FHT-CP is a bit faster than the exact method. 

\begin{figure}
    \begin{tikzpicture}[every mark/.append style={mark size=2pt}]
        \begin{groupplot}[group style = {group size = 3 by 1, horizontal sep = .5em, group name = robustness}, height=4cm, width=.4\textwidth,grid = both, grid style={line width=.1pt, draw=gray!30}, major grid style={line width=.2pt,draw=gray!50}, ylabel style = {yshift=-1.5ex}, xlabel style = {yshift=1.5ex}, ymin = 1, ymax = 10000, xtick = {0, 0.25, 0.5, 0.75, 1}, max space between ticks=20]
            \nextgroupplot[
                xlabel = {},
                ylabel = {QPS (1/s)},
                ymode=log,
                legend style = {opacity = 1.0, text opacity = 1, at = {(2.1, 1.3)}},
                legend columns = 3,
                grid = both, grid style={line width=.1pt, draw=gray!30},
major grid style={line width=.2pt,draw=gray!50},
            ]
\addplot coordinates { (0.197780, 2328.107544) (0.274560, 1177.412316) (0.487530, 347.156600) (0.639110, 199.773429) (0.834920, 54.947634) (0.889600, 31.264120) (0.948600, 15.382934) };
\addlegendentry{CP};
\addplot coordinates { (0.217050, 1627.530098) (0.309070, 637.729792) (0.539980, 207.844760) (0.694440, 118.066237) (0.871780, 30.319863) (0.920850, 18.030915) (0.968930, 9.099753) };
\addlegendentry{FHT-CP};
\addplot coordinates { (0.238340, 1188.359630) (0.352130, 527.994856) (0.619880, 152.211096) (0.778640, 76.867077) (0.923150, 32.138610) (0.956870, 21.026931) (0.985020, 10.550255) };
\addlegendentry{HP};
            \nextgroupplot[
                xlabel = {Recall},
                ylabel = {},
                yticklabels = {, ,},
                ymode=log,
                legend style = {opacity = 1.0, text opacity = 1},
                grid = both, grid style={line width=.1pt, draw=gray!30},
major grid style={line width=.2pt,draw=gray!50},
            ]

\addplot coordinates { (0.272160, 2485.445700) (0.352690, 1931.126076) (0.549150, 979.886461) (0.680150, 570.416917) (0.821350, 263.646418) (0.876740, 184.161734) (0.944490, 102.302105) };
\addplot coordinates { (0.295600, 2837.869323) (0.389160, 1977.281027) (0.613260, 820.133249) (0.734160, 472.841963) (0.870270, 213.582206) (0.914250, 151.447258) (0.964430, 85.402692) };
\addplot coordinates { (0.337310, 1716.191782) (0.442700, 987.146609) (0.676270, 367.686025) (0.807100, 213.014258) (0.927630, 102.810280) (0.958730, 73.589808) (0.985230, 40.636774) };
    \nextgroupplot[
                xlabel = {},
                ylabel = {},
                yticklabels = {, ,},
                ymode=log,
                legend style = {opacity = 1.0, text opacity = 1},
                grid = both, grid style={line width=.1pt, draw=gray!30},
major grid style={line width=.2pt,draw=gray!50},
            ]

\addplot coordinates { (0.261410, 2194.791697) (0.343960, 1789.626751) (0.537790, 908.059069) (0.663820, 612.577444) (0.816580, 361.017337) (0.871700, 259.153172) (0.935870, 154.649420) };
\addplot coordinates { (0.303840, 2479.955647) (0.408640, 1831.641755) (0.634680, 805.553295) (0.747010, 552.533155) (0.875090, 292.949707) (0.918860, 226.234184) (0.968050, 111.696293) };
\addplot coordinates { (0.378260, 860.564333) (0.487540, 596.979718) (0.707630, 271.082372) (0.826060, 168.717916) (0.935260, 87.088706) (0.963630, 63.972241) (0.986620, 36.784454) };
%
\end{groupplot}
\end{tikzpicture}
\caption{Hash function comparison (pooling). Left: 512 MB, center: 1 GB, right: 4 GB.}
\label{plot:eval:functions}
\end{figure}
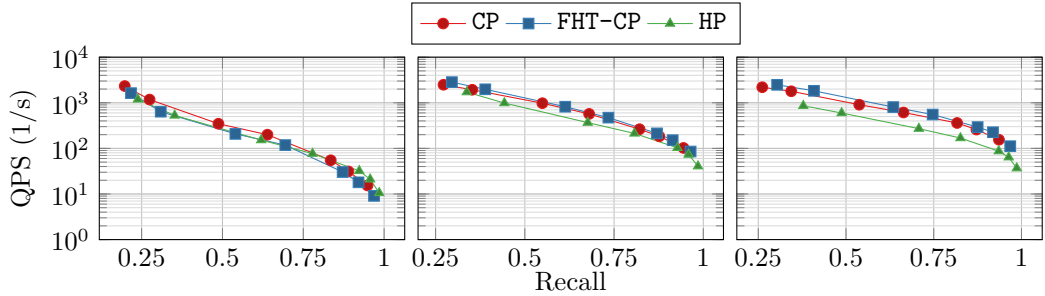

\subsection{Summary of Implementation Choices}
In light of our first high-level question, we observed that sketching provides a good performance increase at negligible cost. 
The larger the index size, the faster PUFFINN can answer queries. 
However, it works well on small index sizes as well. 
We fixed the pooling strategy as the evaluation strategy since it was much faster than tensoring and allows to use exact methods such as CP LSH when FHT-CP does not satisfy the guarantees one wishes for. 

\subsection{Comparison to Other Approaches}
\label{sec:eval:comparison}
\begin{figure}
    \centering
    \begin{tikzpicture}[every mark/.append style={mark size=2pt}]
        \begin{groupplot}[legend style = {font = \tiny}, group style = {group size = 2 by 2, horizontal sep = .5em, vertical sep = 1.5em, group name = robustness}, height=5cm, width=.55\textwidth,grid = both, grid style={line width=.1pt, draw=gray!30}, major grid style={line width=.2pt,draw=gray!50}, ylabel style = {yshift=-1.5ex}, xlabel style = {yshift=1.5ex}, ymin = 1, ymax = 10000, xtick = {0, 0.25, 0.5, 0.75, 1}, max space between ticks=20]
            \nextgroupplot[
            ylabel={ QPS (1/s) },
            ymode = log,
            yticklabel style={/pgf/number format/fixed,
                              /pgf/number format/precision=3},
            legend style = {opacity = 1, text opacity = 1, at = {(2, 1.25)}},
            legend columns = 7,
            ]
\addplot coordinates { (0.245080, 2350.690080) (0.305660, 2160.448635) (0.398840, 1704.567371) (0.417290, 1404.327207) (0.603390, 906.493635) (0.633380, 841.508923) (0.741730, 595.605713) (0.744790, 575.058806) (0.874340, 370.583320) (0.909580, 305.167711) (0.913460, 224.703682) (0.947470, 202.748761) (0.955420, 199.508916) (0.977050, 143.558736) (0.978540, 110.979025) (0.982110, 62.285339) (0.987150, 41.811500) (0.988230, 8.645277) };
\addlegendentry{PUFFINN};
\addplot coordinates { (0.030190, 32962.764066) (0.128310, 27428.164607) (0.237660, 22898.460610) (0.488230, 13199.251973) (0.556670, 8809.843708) (0.693740, 6960.284120) (0.888890, 2201.875900) (0.989490, 253.086339) };
\addlegendentry{ONNG};
\addplot coordinates { (0.349860, 9156.567604) (0.390860, 8879.500658) (0.572370, 5064.716525) (0.613340, 4508.790781) (0.658350, 4051.427864) (0.659210, 3223.624352) (0.701590, 3044.104445) (0.742210, 1852.938265) (0.839070, 1211.320844) (0.868060, 819.513557) (0.886270, 716.574852) (0.901130, 454.280790) (0.912090, 444.737851) (0.931000, 235.218103) (0.940940, 229.100772) (0.965130, 117.361608) };
\addlegendentry{IVF};
\addplot coordinates { (0.164430, 9330.493531) (0.218230, 7559.869322) (0.294580, 5562.916624) (0.317850, 3933.926293) (0.419090, 3233.220822) (0.449080, 2493.920903) (0.523540, 1972.671651) (0.556540, 1640.648394) (0.585020, 1219.677797) (0.627330, 1159.179043) (0.661450, 994.833496) (0.687910, 799.301168) (0.753170, 539.110681) (0.783710, 491.844345) (0.806310, 422.713988) (0.837030, 297.520420) (0.859490, 278.800130) (0.878550, 249.071256) (0.903680, 168.344650) (0.919640, 160.166432) (0.931890, 147.259579) (0.961560, 73.442605) (0.969240, 67.628424) (0.973730, 65.346084) (0.983730, 37.677872) (0.987010, 36.447871) (0.988510, 35.705482) (0.992820, 20.141788) (0.993700, 19.902796) (0.994240, 19.788184) };
\addlegendentry{ANNOY};
\addplot coordinates { (0.121420, 2730.761505) (0.217260, 854.536368) (0.311000, 462.578443) (0.404750, 272.549385) (0.472010, 184.489408) (0.563530, 127.099595) (0.652780, 87.962942) (0.768080, 59.651203) (0.817610, 46.123620) (0.855100, 38.712782) (0.917820, 30.565758) (0.943270, 27.049781) (0.975460, 22.427271) };
\addlegendentry{VPTree(nmslib)};
\addplot coordinates { (0.003810, 43975.268993) (0.007850, 25073.628982) (0.022730, 15771.416623) (0.043190, 11083.460323) (0.074800, 6316.824504) (0.081380, 5666.634650) (0.126470, 2564.269614) (0.170950, 2505.794941) (0.280080, 1338.614891) (0.429620, 690.128143) (0.600470, 354.388338) (0.696540, 237.373764) (0.721880, 217.380561) (0.746770, 191.110653) (0.790960, 160.233260) (0.831150, 130.348670) (0.851760, 112.355403) (0.866230, 98.124183) (0.884400, 90.262786) (0.897190, 66.920746) (0.925200, 58.423953) (0.948560, 50.674622) (0.955810, 45.285885) (0.964820, 41.197608) (0.971690, 40.087392) };
\addlegendentry{FALCONN};
            \addplot coordinates { (0.406390, 83.820759) (0.600470, 64.658111) (0.714960, 49.408357) (0.812810, 45.129181) (0.884680, 35.969848) (0.927640, 29.743923) };
\addlegendentry{FLANN};

    \nextgroupplot[
                ylabel = {},
                ymode=log,
                legend style = {opacity = 0.6, text opacity = 1},
                yticklabels = {, ,},
            ]
            \addplot coordinates { (0.358660, 270.986506) (0.365370, 260.913404) (0.460110, 239.351683) (0.472610, 206.107058) (0.484740, 180.440447) (0.499630, 168.957223) (0.558210, 160.804489) (0.562020, 141.096667) (0.643810, 119.941349) (0.714130, 94.886561) (0.714490, 77.700017) (0.751900, 77.419652) (0.761750, 73.514177) (0.824780, 62.219760) (0.864060, 45.074552) (0.887470, 38.902288) (0.892750, 32.274077) (0.902630, 27.849255) (0.912080, 26.005760) (0.944070, 20.798543) (0.956940, 18.458453) (0.972840, 15.260163) (0.975090, 12.734050) (0.979730, 9.955602) (0.983370, 9.648467) (0.985120, 8.517598) (0.988040, 7.896482) (0.989000, 6.461813) };
\addplot coordinates { (0.131490, 22968.303779) (0.236520, 20855.335263) (0.337440, 15045.664275) (0.645720, 7140.446306) (0.726610, 5201.759463) (0.894050, 2873.289267) (0.987180, 870.854586) (0.995450, 26.997083) };
\addplot coordinates { (0.565030, 2960.373800) (0.592570, 2020.101597) (0.618340, 1420.979473) (0.752390, 1296.812646) (0.774020, 1290.796422) (0.815010, 1031.502691) (0.832520, 796.742088) (0.840090, 563.384564) (0.900410, 348.324086) (0.907250, 201.895021) (0.922970, 193.349818) (0.941090, 115.758326) (0.949190, 64.096256) (0.959570, 38.001092) (0.967400, 18.436840) (0.969450, 13.133526) (0.984930, 9.613749) (0.996570, 4.936213) (0.998850, 4.443679) };
\addplot coordinates { (0.166900, 2335.247787) (0.178960, 1976.076091) (0.192740, 1643.835001) (0.218750, 1535.014975) (0.238350, 1310.856314) (0.294120, 987.557326) (0.324230, 905.686679) (0.367160, 671.112863) (0.408460, 664.383248) (0.502320, 436.807045) (0.543530, 288.652458) (0.626610, 234.592395) (0.671910, 171.971544) (0.712670, 137.590431) (0.757320, 109.611675) (0.775490, 81.461524) (0.829620, 67.887422) (0.883220, 33.108341) (0.909170, 18.612013) (0.938780, 10.430628) };
\addplot coordinates { (0.137410, 3760.043734) (0.217950, 370.379847) (0.345580, 84.013709) (0.499250, 35.257321) (0.639950, 18.531209) (0.773300, 10.480108) (0.868730, 7.062979) (0.934210, 5.152528) (0.963890, 4.590523) (0.980740, 4.407078) (0.993670, 4.120175) (0.996200, 4.080274) (0.998610, 3.647428) };
\addplot coordinates { (0.102080, 12529.069273) (0.103030, 11775.768727) (0.109540, 8260.154133) (0.120860, 2948.676704) (0.125300, 2849.823317) (0.141620, 1919.103602) (0.147400, 1485.017262) (0.170030, 715.109555) (0.193200, 684.689188) (0.248070, 368.028373) (0.338950, 159.491285) (0.451200, 73.763625) (0.472550, 65.166001) (0.551600, 58.954447) (0.581160, 51.720950) (0.621540, 43.588488) (0.667350, 34.956270) (0.715490, 33.383502) (0.741430, 28.012388) (0.769110, 17.520289) (0.814060, 13.750133) (0.854170, 13.173451) (0.875870, 12.829132) (0.902830, 11.053083) };
    \nextgroupplot[
                xlabel = {Recall},
                ylabel = {QPS (1/s)},
                ymode=log,
                yticklabel style={/pgf/number format/fixed,
                              /pgf/number format/precision=3},
                legend style = {opacity = 0.6, text opacity = 1},
            ]
            \addplot coordinates { (0.317550, 796.695312) (0.366750, 704.584591) (0.385430, 551.546248) (0.433050, 464.718085) (0.474750, 441.033640) (0.512210, 415.041057) (0.540720, 340.818162) (0.557800, 329.347392) (0.593640, 281.864355) (0.613120, 211.029767) (0.644860, 205.819089) (0.737490, 202.765856) (0.761490, 177.992666) (0.774980, 155.093396) (0.781090, 148.906942) (0.821520, 132.183930) (0.832530, 129.818111) (0.849110, 125.522810) (0.868190, 113.813982) (0.869630, 111.201873) (0.881460, 105.495611) (0.887320, 96.480588) (0.890180, 91.338161) (0.917370, 73.115608) (0.953280, 52.812854) (0.963990, 31.674110) (0.971600, 29.975617) (0.972870, 25.733120) (0.978730, 25.189247) (0.979230, 23.370514) (0.981380, 22.363807) (0.984630, 19.236179) (0.985990, 18.907295) };
\addplot coordinates { (0.179920, 16049.137250) (0.253440, 13091.592357) (0.371220, 10971.096598) (0.761010, 5702.031151) (0.841410, 5023.981193) (0.956900, 4504.044478) (0.991220, 1041.345437) (0.995750, 170.370199) };
            \addplot coordinates { (0.537270, 2660.996829) (0.579510, 1968.293189) (0.801900, 1258.983180) (0.857190, 917.348178) (0.877180, 808.319733) (0.897130, 474.684879) (0.914210, 243.798380) (0.971820, 154.996579) (0.977760, 106.226679) (0.982770, 90.499913) (0.986560, 50.067590) (0.987320, 29.255316) (0.990020, 29.207692) (0.992680, 15.272680) };
\addplot coordinates { (0.289280, 4016.837923) (0.318740, 3383.983332) (0.410930, 2686.684839) (0.434630, 2002.742914) (0.552950, 1676.551452) (0.579210, 1225.088316) (0.665200, 1063.616527) (0.691500, 776.773885) (0.710960, 660.655445) (0.763720, 620.856499) (0.786810, 507.927189) (0.801160, 453.261406) (0.860830, 311.909069) (0.876130, 267.386915) (0.884800, 250.808450) (0.908690, 178.238151) (0.920320, 170.204637) (0.926310, 157.829053) (0.949320, 101.093449) (0.954240, 96.227566) (0.968680, 47.484592) (0.974720, 46.685410) (0.979340, 26.761212) (0.982520, 25.683654) (0.987430, 14.795904) (0.988070, 14.287293) };
\addplot coordinates { (0.178430, 2661.590270) (0.220930, 1133.162904) (0.271360, 660.322342) (0.314980, 410.419409) (0.362870, 251.843414) (0.424470, 141.083393) (0.482700, 91.161525) (0.560880, 51.686727) (0.614410, 42.836844) (0.673420, 29.431715) (0.763900, 18.979686) (0.799680, 13.252016) (0.884200, 7.640512) };
\addplot coordinates { (0.103080, 15691.801881) (0.110450, 5344.252485) (0.120950, 3973.312731) (0.138730, 2382.918200) (0.162450, 1316.004218) (0.175320, 962.572735) (0.243500, 548.491913) (0.350070, 258.548466) (0.511200, 115.793667) (0.671630, 68.967083) (0.767140, 50.187023) (0.795280, 43.319448) (0.810480, 27.866077) (0.839490, 27.650102) (0.885670, 27.522016) (0.900520, 26.459888) (0.927890, 20.264033) (0.933320, 13.734845) (0.954310, 13.277595) (0.969670, 9.624735) (0.975660, 8.764131) (0.979300, 7.836705) (0.984010, 7.018094) };
    \nextgroupplot[
                xlabel = {Recall},
                ylabel = {},
                ymode=log,
                legend style = {opacity = 0.6, text opacity = 1},
                yticklabels = {, ,},
            ]
\addplot coordinates { (0.605000, 748.519862) (0.820000, 514.856424) (0.971000, 360.306134) (0.978000, 243.791433) (0.992000, 178.238772) (0.995000, 156.441904) (0.999000, 54.631678) };
\addplot coordinates { (0.019000, 87.964881) (0.940000, 6.867624) };
            \addplot coordinates { (0.001000, 4665.236314) (0.002000, 3390.229805) (0.003000, 2077.375346) (0.004000, 1504.388433) (0.005000, 1337.468555) (0.013000, 705.752269) (0.033000, 446.253319) (0.050000, 278.018691) (0.086000, 145.871374) (0.116000, 80.488105) (0.220000, 41.135089) };
\addplot coordinates { (0.001000, 3016.658064) (0.002000, 1664.555680) (0.004000, 970.893827) (0.005000, 577.424402) (0.006000, 398.194648) (0.010000, 259.250064) (0.012000, 240.236657) (0.014000, 198.830419) (0.021000, 146.101377) (0.027000, 123.370642) (0.044000, 80.196532) (0.111000, 35.120662) (0.208000, 18.751385) (0.210000, 18.564438) (0.355000, 10.540757) (0.366000, 10.481671) };
\addplot coordinates { (0.135000, 65.325733) (0.233000, 37.555880) (0.456000, 17.093406) (0.662000, 11.985236) (0.822000, 9.593139) (0.947000, 7.749562) (0.966000, 6.964156) (0.988000, 6.623116) (0.999000, 6.186428) };
\addplot coordinates { (0.011000, 25387.559165) (0.014000, 19165.024327) (0.027000, 12432.540783) (0.098000, 6482.274799) (0.159000, 3226.387692) (0.186000, 1447.871443) (0.210000, 1373.696642) (0.354000, 725.944586) (0.487000, 504.981400) (0.541000, 365.217842) (0.611000, 309.996201) (0.658000, 297.629403) (0.686000, 203.940614) (0.731000, 182.589342) (0.783000, 161.743512) (0.834000, 154.327609) (0.862000, 120.896137) (0.941000, 100.710643) (0.949000, 75.754348) };
\end{groupplot}
\end{tikzpicture}
\caption{ Comparison of different implementations with PUFFINN (index size at most 8 GB). Top left: \textsf{Glove-1M}, top right: \textsf{Glove-2M}, bottom left: \textsf{GNEWS-3M}, bottom right: \textsf{Synthetic}. }
    \label{plot:comparison}
\end{figure}
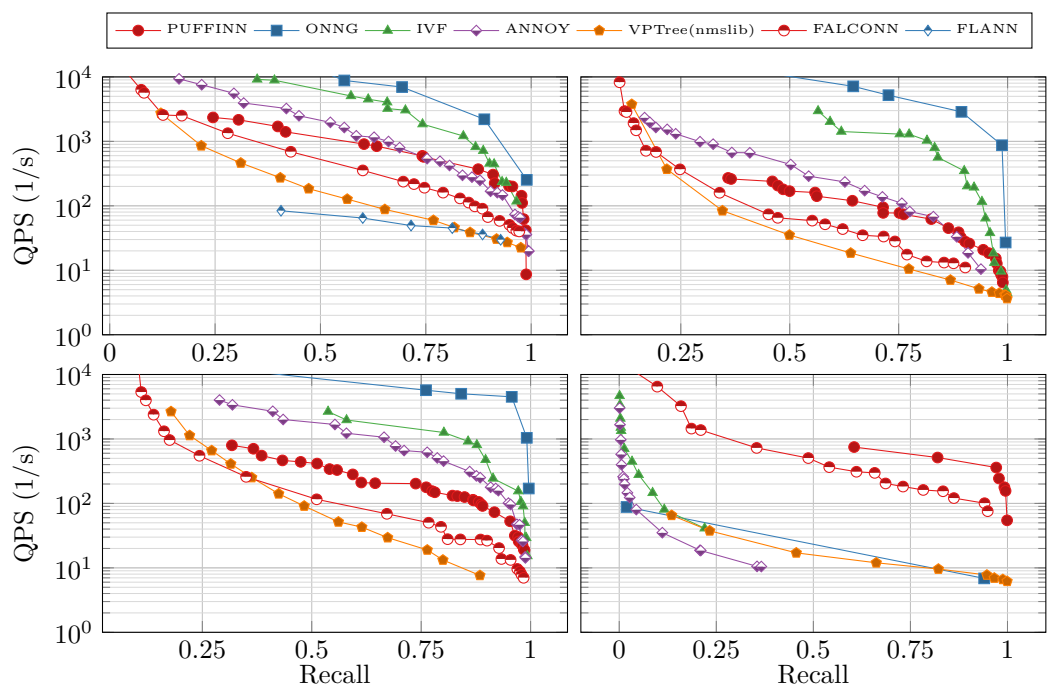

We turn our focus on comparing PUFFINN to other approaches on real-world and synthetic datasets (HL-Q2). 
Figure~\ref{plot:comparison} gives an overview of the performance-quality trade-off achieved by state-of-the-art $k$-NN approaches on three real-world dataset and a synthetic one.
Among the implementations that support automatic parameter tuning, PUFFINN is usually by a large factor the fastest implementation. 
We remark that the automatic parameter tuning of FLANN failed to build an index within 6 hours on three out of the four datasets and was disregarded on those.

PUFFINN managed to obtain at least recall .95 on every dataset, whereas the tuning of the VPTree failed to achieve high recall on \textsf{GNEWS-3M}.
PUFFINN shows better performance than FALCONN for most of the performance-quality space. It is comparable in performance to ANNOY on most of the real-world datasets, and is particular competitive in the high-recall setting.
On real-world datasets, IVF and in particular ONNG show better performance than PUFFINN except for very high recall.
This is an indicator that graph-based approaches perform best on these real-world datasets (for good manual parameter choices).

On the synthetic dataset, only LSH-based approaches achieve high recall at high QPS. 
VPTree and ONNG manage recall close to 1, but are more than a factor 10 slower than PUFFINN. IVF and ANNOY
fail to achieve recall higher than 40\% on the synthetic dataset.

To come back to our second high-level question: It is possible to compete with state-of-the-art implementations of $k$-NN using a parameterless method with guarantees. 
PUFFINN is easy to use and achieves performance comparable to many of its competitors. 
Among LSH variants, it is as fast as FALCONN which does not give guarantees.  This means that our engineering choices allowed an LSH implementation with theoretical guarantees that come ``for free''.


\bibliographystyle{plainurl}
\bibliography{lit}

\conf{}
{
\appendix
\section{The LSH Framework}
\label{app:lsh}
Having access to an LSH family $\HH$ as defined in Section~\ref{sec:prelim} allows us to build a data structure with the following properties solving the $(c, r)$-near neighbor problem.

\begin{theorem}[{\cite[Theorem 3.4]{HarPeledIM12}}]
   Suppose that for some metric space $(X, \dist)$ and some factor $c > 1$, there exists an LSH family such that
   $\Pr[h(q)=h(x)] \ge p_1$ when $\text{dist}(q,x)\le r$ and
   $\Pr[h(q)=h(x)] \le p_2$ when $\text{dist}(q,x)\ge cr$ with $p_1 > p_2$.
   Let $S \subseteq X$ be a dataset containing $n$ points. 
   Then there exists a data structure for $S$ such that for a given query~$q$, it returns with constant probability a point within distance $cr$ in $S$, if there exists a point within distance~$r$ in $S$.
   The algorithm uses $O(d n + n^{1 + \rho})$ space and evaluates $O(n^\rho)$ hash functions per query, where $\rho = \frac{\log (1/p_1)}{\log (1/p_2)}$.
   \label{thm:lsh}
\end{theorem}

\begin{proof}
Given access to an LSH family $\HH$ with the properties stated in the theorem
and two parameters~$L$ and $k$ (to be specified below), repeat the following process independently for
each $i$ in $\{1, \ldots, L\}$:
Choose $k$ hash functions $g_{i, 1},\ldots,g_{i, k}$ independently at random from~$\HH$.
For each point $p \in S$, we view the
sequence $h_i(p) = (g_{i, 1}(p),\ldots, g_{i, k}(p)) \in R^k$ as the hash code of $p$,
identify this hash code with a bucket in a table, and store a reference to $p$
in bucket $h_i(p)$.
To avoid storing empty buckets from $R^k$, we resort to
hashing and build a hash table $T_i$ to store the non-empty buckets for $S$ and $h_i$.

Given a query $q \in X$, we retrieve all points from the buckets $h_1(q),
\ldots, h_L(q)$ in tables $T_1, \ldots, T_L$, respectively, and report
a close point in distance at most $cr$ as soon as we find such a point.
Note
that the algorithm stops and reports that no close points exists after
retrieving more than $3L$ points, which is crucial to guarantee query time
$O(n^\rho)$.

The parameters $k$ and $L$ are set according to the following reasoning.
First,
set $k$ such that it is expected that at most one distant point at distance at
least $cr$ collides with the query in one of the repetitions.
This means that
we require $n
p_2^k \leq 1$ and hence we define $k = \lceil \frac{\log n}{\log(1/p_2)} \rceil$.
To find a close point at distance at most $r$ with probability at
least $1 - \delta$, the number of repetitions $L$ must satisfy
$\delta \leq (1 - p_1^k)^L \leq \exp(-p_1^k \cdot L)$.
This means that $L$
should be at least $p_1^{-k} \ln \delta$ and simplifying yields $L = O(n^\rho)$.
Note that these parameters are set to work even in a worst-case scenario where there is
exactly one point at distance $r$ and all other points have distance slightly larger 
than $cr$.
\end{proof}

\section{Missing Details of Data Structure and Its Theoretical Properties}
\label{app:ds}

\subsection{Discussion and Proof of Lemma~\ref{lem:knn_time}}

Lemma \ref{lem:knn_time} states that with probability at least $1 - \delta$ the expected running time is bounded by the time it would take using the natural algorithm to find the entire $k$NN set with probability at least $\delta$ in addition to the time it takes to traverse the $L$ LSH tries bottom up and to report $k$ points from each.
The bound is very conservative since the guarantees have to cover adverse instances.
Consider for example the case where there is a large gap in the collision probability between the $k$th and the $k+1$st nearest neighbor.
In this case the algorithm would have to find the entire $k$NN set before the stopping criteria matches that of the natural algorithm. 
We could also have the problem that the $k-1$ nearest neighbors all collide with the query with probability 1 so that the query algorithm encounters all of these points in every LSH trie.
\begin{proof}
	Let $i'$ denote the largest $i$ such that $j' = \lceil \ln(k/\delta)/p(q, x_k)^i \rceil \geq L$.
	If we denote by $i^*$ the choice of level underlying $OPT(L, K, k, \delta/k)$ then we have that $i' \geq i^*$.
	We will proceed by bounding the probability that \textsc{adaptive-kNN}$(q, k, \delta)$ proceeds beyond step $(i', j')$. 
	The probability of having missed a point in the true top-$k$ at step $(i', j')$ is at most $\delta/k$ and by a union bound the probability of not having found all of them is at most $\delta$. 
	Furthermore, in the event of having found the true top-$k$ at step $(i', j')$ the stopping criteria takes effect.
	Since $\delta \leq 1/2$ the expected running time conditioned on stopping no later than at step $(i', j')$ is no more than a constant times the expected running time of searching to step $(i' , j')$.
	The expected running time of searching $j'$ trees at depth $i'$ can be upper bounded as follows:
	\begin{align*}
		&\frac{\ln(\delta/k)}{p(q, x_k)^{i'}}(i' + \sum_{s = 1}^{k} p(q, x_s)^{i'} + \sum_{s = k+1}^{n} p(q, x_s)^{i'})  
		\leq L(K + k) + \frac{\ln(\delta/k)}{p(q, x_k)^{i'}} \sum_{s = k+1}^{n} p(q, x_s)^{i'} \\
		&\quad \leq L(K + k) + \frac{\ln(\delta/k)}{p(q, x_k)^{i^*}} \sum_{s = k+1}^{n} p(q, x_s)^{i^*} \leq L(K +k) + OPT(L, K, k, \delta/k).
	\end{align*}
\end{proof}

\subsection{Reducing the number of LSH evaluations}
Many locality-sensitive hash functions have an evaluation time that is at least linear in the description size of the input.
To reduce the contribution from LSH evaluations to the query time we can reuse a single LSH function in multiple LSH tries so we only have to pay once to evaluate it. 
The LSH literature contains two different techniques to achieve this goal: tensoring~\cite{andoni2006efficient} and pooling~\cite{Dahlgaard2017FastSS}.
See~\cite{Christiani17} for a more detailed description in relation to the approximate near neighbor problem.
We proceed by describing how tensoring and pooling can be used in the context of solving the $k$NN problem using the approach of Algorithm~\ref{alg:adaptiveknn}.
Tensoring has the desirable property that it does not introduce additional parameters to the algorithm, whereas pooling requires specifying the pool size $m$ in advance.

\subparagraph*{Tensoring}
The variant of tensoring described here will reduce the number of LSH functions evaluated from $LK$ to $\sqrt{L}K$. 
The cost of using fewer independent LSH evaluations is that we need to perform more lookups and distance computations in order to achieve the same correctness guarantees.
However, for constant failure probability, say, $\delta = 1/2$, we only have to increase the number of buckets that we search by a small constant factor compared to using independent LSH functions.

Assume that $L$ is an even power of two and that $K$ is an even positive integer. 
We form two collections that each hold $\sqrt{L}$ tuples of $K/2$ independent LSH functions.
Denote the $s$th hashfunction in the $t$th tuple of the first and second collection by $h_{s,t}$ and $h^{'}_{s,t}$, respectively.
Each of the $L$ LSH tries is now indexed by $j_1, j_2 \in \{1, \dots, \sqrt{L}\}$ which denotes the tuple of LSH functions used from the first and second collection, respectively.
For $i \in \{0, 2, 4, \dots, K \}$ the set of points in the bucket associated with a query point $q$ at depth $i$ in the LSH trie indexed by $j_1, j_2$ is given by:
\begin{align*}
	S_{i,j_1, j_2}(q) &= \{ x \in P \mid h_{1,j_1}(x) = h_{1,j_1}(q) \land \dots \land h_{i/2,j_1}(x) = h_{i/2,j_1}(q) \\ 
	&\qquad \land h^{'}_{1,j_2}(x) = h^{'}_{1,j_2}(q) \land \dots \land h^{'}_{i/2,j_2}(x) = h^{'}_{i/2, j_2}(q) \}. 
\end{align*}
The query algorithm under the tensoring scheme (Algorithm \ref{alg:knn_tensoring}) proceeds by searching through the LSH tries at depth $i = K, K-2, \dots, 0$.
At each depth $i$ for $m = 1, \dots, \sqrt{L}$ we gradually search the $m^2$ buckets that are formed by pairing up the first $i/2$ hash functions from the first $m$ tuples from the two collections.

We now consider the stopping criteria.
The probability that the query point does not collide with the $k$NN in the first $i/2$ hash functions of the first $m$ tuples in either collection is given by $1 - p(q, x_k)^{i/2} \leq 1 - p(q, x_k')^{i/2}$.
The algorithm fails to find the $k$NN if there is not at least one such collision in each collection and the stopping criteria follows from a union bound.
\begin{algorithm}
\SetKw{or}{ or }
\DontPrintSemicolon
Initialize PQ with a capacity of k unique keys \; 
\For{$i \gets K, K - 2, K - 4, \dots, 0$}{
	\For{$m \gets 1, 2, \dots, \sqrt{L}$}{
		\For{$j \gets 1, 2, \dots, m$} {
			\lFor{$x \in \Pi_{i,j,m}(q)$}{PQ.insert$(x, \dist(q, x))$}
			\lFor{$x \in \Pi_{i,m,j}(q)$}{PQ.insert$(x, \dist(q, x))$}
		}
		\lIf{$i = 0 \or 2(1 - p(q, x'_k)^{i/2})^m \leq \delta$}{\Return PQ}
	}
}
\caption{\textsc{adaptive-kNN-tensoring}$(q, k, \delta)$} \label{alg:knn_tensoring}
\end{algorithm}
\subparagraph*{Pooling}
Another approach to reduce the number of independent LSH functions that have to be stored and evaluated is to form a pool of $m$ LSH functions.
The $K$ hash functions used for each of the $L$ tries are selected independently as random subsets of size $K$ from the pool of hash functions. 
In algorithm \ref{alg:knn_pooling} we let $h_{i,j}$ denote the $i$th hash function from the $j$th random subset from the pool of hash functions.
By adapting the analysis from~\cite{Christiani17} we obtain a failure probability of $\delta \leq 1/2$ by setting the stopping criteria to ensure two things:
First, the expected number of times that the $k$-nearest neighbor has been found should be at least two. 
And secondly, the pool size $m$ should be large enough to ensure that the variance in the number of collisions between the query point and the $k$NN is sufficiently low for Cantelli's inequality to give the desired failure probability. 
\begin{algorithm}
\SetKw{or}{ or }
\SetKw{and}{ and }
\DontPrintSemicolon
Initialize PQ with a capacity of k unique keys \; 
\For{$i \gets K, K - 1, \dots, 0$}{
	\For{$j \gets 1, 2, \dots, L$}{
		\lFor{$x \in \Pi_{i,j}(q)$}{PQ.insert$(x, \dist(q, x))$}
		\lIf{$i = 0 \or j \cdot p(q, x'_k)^i \geq 2 \and m \geq 5i^{2} / p(q, x'_k)$}{\Return PQ}
	}
}
\caption{\textsc{adaptive-kNN-pooling}$(q, k)$} \label{alg:knn_pooling}
\end{algorithm}
For completeness we include the proof of correctness that follows the approach of Christiani~\cite{Christiani17}.
\begin{proof}
Cantelli's inequality states that for a non-negative random variable $X$ with mean $\mu$ and variance $\sigma^2$ we have that $\Pr[X \leq 0] \leq \sigma^2 / (\mu^2 + \sigma^2)$.
Let $X$ denote the number of collisions between the query point and the $k$NN $x_k$ when searching $j$ tries at depth $i$ under the pooling scheme.
We have that $X = \sum_{t = 1}^{j} X_t$ is the sum of collisions over the $j$ tries and that $X_t = \Pi_{s = 1}^{i}1\{h_{s,t}(q) = h_{s,t}(x_k)\}$ is an indicator variable for whether a set of $i$ random hash functions from the pool all produce a collision between $q$ and $x_k$.
By linearity of expectation and because the LSH functions $h_{s,t}$ used for a specific $X_t$ are independent we have that $\mu = \E[X] = jp(q, x_k)^i$.
To bound the variance $\sigma^2 = \E[X^2] - \mu^2$ we proceed by upper bounding $\E[X^2] \leq \mu + j^2 \E[X_1 X_2]$.
We will use the principle of deferred decisions to upper bound $\E[X_1 X_2] = \E[\Pi_{s = 1}^{i}1\{h_{s,1}(q) = h_{s,1}(x_k)\}\Pi_{s = 1}^{i}1\{h_{s,2}(q) = h_{s,2}(x_k)\}]$. 
First a set of $i$ indices for random hash functions from the pool are sampled for $X_1$.
Secondly, each hash function sampled for $X_2$ either collides with a hash function sampled for $X_1$, or a new random hash function is sampled.
The first event happens with probability at most $i/(m - i)$, resulting in the bound $\E[X_1 X_2] \leq p(q, x_k)^i (i/(m-i) + (1-i/(m-i)) p(q, x_k))^i$.
If we assume that $m \geq 2i$ we can further rewrite the upper bound as $\E[X_1 X_2] \leq p^{2i} \exp(2i/mp)$.
We obtain a failure probability $\delta \leq 1/2$ from Cantelli's inequality by setting $\mu \geq 2$ and $m \geq 5i^{2} / p(q, x_k)$ which we use as stopping criteria in Algorithm~\ref{alg:knn_pooling}.
\end{proof}

\section{Choice of Segment Size $B$}
\label{app:segment:size}
\begin{table}[tbhp]
\centering
\begin{tabular}{llllll}
\toprule
$B$ & 1 & 4 & 8 & 12 & 16\\
\midrule
Avg. Recall & 0.936 & 0.922 & 0.929 & 0.932 & 0.932\\
QPS & 35.0 & 39.7 & 38.7 & 39.4 & 34.6\\
\bottomrule
\end{tabular}
\caption{\label{tab:segment_size}Query performance using different choices of $B$}
\end{table}
When retrieving additional candidates by reducing the length of the prefix, we retrieve candidates in segments of size $B$, even if only the first candidate's prefix matches. This reduces the number of random memory accesses, which improves the running time even though the number of candidates increases. Some measurements for \textsf{Glove-1M}, 2 GB index size, and recall .95 are shown in Table~\ref{tab:segment_size} for different values of $B$, where $B = 1$ is the case where only points with a matching prefix are considered. Based on these results, we chose to use $B = 12$.

}

\end{document}